\documentclass[twoside]{article}
\usepackage[accepted]{aistats2e}

\usepackage[authoryear, round]{natbib}
\usepackage{amssymb, amsmath, amsthm}
\usepackage[dvips]{graphicx}

\newcommand{\Unif}{{\rm Unif}}
\newcommand{\I}{{\rm I}}
\newcommand{\Normal}{{\rm Normal}}
\newcommand{\E}{\mathbb{E}}
\renewcommand{\v}{{\rm var}}

\renewcommand{\P}{\mathbb{P}}

\newtheorem{lemma}{Lemma}[section]

\begin{document}

\twocolumn[

\aistatstitle{A simple sketching algorithm for entropy estimation over streaming data}

\aistatsauthor{ Peter Clifford \And Ioana Ada Cosma }
\aistatsaddress{ Department of Statistics \\
University of Oxford \\
Oxford, UK. OX1 3TG  \\
{\tt peter.clifford@jesus.ox.ac.uk} \And Department of Mathematics and Statistics \\
University of Ottawa \\
Ottawa, ON, Canada. K1N 6N5 \\
{\tt icosma@uottawa.ca}}
]

\begin{abstract}
We consider the problem of approximating the empirical Shannon entropy of a high-frequency data stream under the relaxed strict-turnstile model, when space limitations make exact computation infeasible. An equivalent measure of entropy is the R\'{e}nyi entropy that depends on a constant $\alpha$.  This quantity can be estimated efficiently and unbiasedly from a low-dimensional synopsis called an $\alpha$-stable data sketch via the method of compressed counting.  An approximation to the Shannon entropy can be obtained from the R\'enyi entropy by taking $\alpha$ sufficiently close to 1. However, practical guidelines for parameter calibration with respect to $\alpha$ are lacking. We avoid this problem by showing that the random variables used in estimating the R\'{e}nyi entropy can be transformed to have a proper distributional limit as $\alpha$ approaches 1: the maximally skewed, strictly stable distribution with $\alpha = 1$ defined on the entire real line.  We propose a family of asymptotically unbiased \emph{log-mean estimators} of the Shannon entropy, indexed by a constant $\zeta > 0$, that can be computed in a single-pass algorithm to provide an additive approximation. We recommend the log-mean estimator with $\zeta = 1$ that has exponentially decreasing tail bounds on the error probability, asymptotic relative efficiency of 0.932, and near-optimal computational complexity.  
\end{abstract}

\section{INTRODUCTION}

Streaming data is ubiquitous in a wide range of areas from engineering, and information technology, finance, and commerce, to atmospheric physics, and earth sciences \citep{Muthukrishnan.05, Aggarwal.07}.  The term \emph{streaming data} refers to the situation where data is continuously generated at high speed, and must be processed in real time to facilitate data analysis and decision making.  The Shannon entropy \citep{SW.49} provides an important characterization of a data stream with many areas of application, e.g., network traffic monitoring for the purpose of anomaly detection or traffic clustering, analysis of commercial search logs, and signal processing.  In network traffic monitoring, \cite{LSOXZ.06} show that the empirical Shannon entropy is an appropriate summary statistic for capturing changes in the underlying traffic distribution.  Changes in the distribution of the number of packets observed at different ports can be indicative of port scanning attacks.

In recent years, several algorithms have been developed for estimating the Shannon entropy over streaming data \citep{BG.06, CDBM.06, ZLOSWX.07, HNO.08.2, CCM.10, LZ.11}. Many of these algorithms are based on the approach of $\alpha$-stable \emph{data sketching} \citep{Indyk.06}.  Sketches are low-dimensional data structures, usually in vector or matrix format; when a new element in the stream is observed, the sketch is potentially updated, and the element is discarded. This update is handled in the same way, irrespective of the order of arrival of past data.  The idea is to construct and maintain on-the-fly a compact synopsis of the data stream such that summary statistics of interest can be accurately approximated from the synopsis; in general, synopsis construction is specific to the statistic of interest.

This paper considers the problem of estimating the Shannon entropy of a data stream under the assumption that the number of distinct elements observed in the stream is prohibitively large, so that the vector of cumulative quantities cannot be stored on main computer memory for fast and efficient access.  We employ the method of $\alpha$-stable data sketching, i.e., transforming distinct stream elements online to distinct realizations of a stable variable of index $\alpha$ (called $\alpha$-stable hereafter), and storing weighted linear combinations of these realizations, independently replicated $k$ times.  These weighted linear combinations, known as \emph{random projections}, form a $k$-dimensional synopsis of the data stream, called a data sketch hereafter, where $k$ is determined by the accuracy desired in approximating the entropy.  We consider data streams under the relaxed strict-turnstile model, which allows deletions, provided that, whenever the entropy is estimated, all cumulative quantities are non-negative.  

We sketch distinct stream elements to pseudo-random variates following the maximally skewed $\alpha$-stable distribution with $\alpha =1$ via the \emph{method of seeding}; the stream elements effectively index the random variates.  This is in contrast to existing approaches that involve sketching with $\alpha$ close to 1, thus introducing an additional source of error as explained in Section~\ref{sec:related}.  We present a family of \emph{log-mean estimators} of the Shannon entropy whose construction is simple and direct, requiring a single pass over the data stream.  We give explicit algorithms for implementing this estimation procedure, and analyze their computational complexity in terms of the length of the stream, and the additive approximation error.

\subsection{Notation and terminology}

A {\it data stream} $S_T$ of length $T$ is a transiently observed sequence of data elements $(i_t,d_t)$ that arrive unordered, with repetition, and at very high rate of transmission.  The item type $i_t$ belongs to a large or possibly infinite set $\mathcal{D}=\{c_1,c_2,\dots,c_N\}$ and the associated quantity is $d_t \in \mathbb{R}$, for $t=1,2,\dots,T$.  If $d_t < 0$, then the data element $(i_t, d_t)$ is a deletion from the stream; otherwise, it is an insertion.  For simplicity, we assume that $T \geq N$.  The empirical probability distribution is given by 
$$p_j = \frac{a_j}{\sum_{i=1}^N a_i}, \ j=1,\ldots,N,
$$
where $a_j = \sum_{t=1}^T d_t \mathbb{I} (i_t = c_j)$ is the cumulative quantity of elements of type $c_j$ at stage $T$, and $a_j \geq 0 \ \forall j$ at every stage $T$ of interest.   So, the empirical distribution is well-defined.  This is called the relaxed strict-turnstile model.

The {\it empirical Shannon entropy} is defined by 
\begin{equation} \label{def:Shannon}
H(p) = -\sum_{j=1}^N p_j \log p_j,
\end{equation} 
where, by convention, $p \log p$ is defined to be 0 when $p=0$, and $\log$ is the logarithm function to the base $e$.  Equivalent measures of entropy are the R\'{e}nyi  \citep{Renyi.61} and Tsallis \citep{Tsallis.88} entropies, given, respectively, by
\begin{align*} 
H_{\alpha} (p) &= \frac{1}{1-\alpha} \log \left ( \sum_{j=1}^N p_j^{\alpha} \right ), \\
S_{\alpha}(p) &= \frac{1}{1- \alpha} \left (\sum_{j=1}^N p_j^{\alpha} -1 \right ),
\end{align*}
for $0 \leq \alpha$ and $\alpha \neq 1$.  $H_{\alpha}(p)$ and $S_{\alpha}(p)$ equal the Shannon entropy in the limit as $\alpha$ tends to 1.  Both quantities $H_{\alpha}(p)$ and $S_{\alpha}(p)$ are functions of the $\alpha$th frequency moment, denoted by $F_{\alpha}(p)$, and defined as
$$
F_{\alpha}(p) = \sum_{j=1}^N p_j^{\alpha},
$$
a connection that is exploited by many algorithms for estimating the Shannon entropy, as explained in Section~\ref{sec:related}.

\subsection{Data sketching and the stable distribution} \label{sec:stable}

We employ the method of data sketching to the $\alpha$-stable distribution.  Following \cite{Zolotarev.86}, the stable distribution has four parameters: index $\alpha \in (0,2]$, skewness $\beta \in [-1,1]$, location $\delta \in \mathbb{R}$, and scale $\gamma > 0$, denoted by $F(x; \alpha, \beta, \gamma, \delta)$.  If $X$ has distribution $F(x; \alpha, \beta, \gamma, \delta)$ (written as: $X \sim F(x; \alpha, \beta, \gamma, \delta)$), then its characteristic function (c.f.) $\phi(\theta) = \mathbb{E} \exp(i \theta X)$, $\theta \in \mathbb{R}$, is given by
\begin{displaymath}
\phi(\theta)  = \left \{ \begin{array}{l}
\exp \left (\gamma^{\alpha} [-|\theta|^{\alpha} + i\theta |\theta|^{\alpha - 1}\beta \tan( \frac{\pi \alpha}{2}) ] + i \delta \theta \right ), \\
\quad \quad \textrm{if $\alpha \neq 1$} \\  
\exp \left (\gamma [-|\theta|- i\theta \beta (\frac{2}{\pi}) \log |\theta| ] + i \delta \theta \right ), \\
\quad \quad \textrm{if $\alpha = 1$,} 
\end{array} \right.
\end{displaymath}
where $\mathbb{E}$ denotes expected value, and $i=\sqrt{-1}$.  If $\beta = \pm 1$, the distribution is called maximally skewed.  In particular, we sketch to the maximally skewed distribution $F(x;1,-1,\pi/2,0)$ by simulating independent draws using the algorithm in Table~\ref{alg:sim} \citep{Zolotarev.86}.

\begin{table}[ht]
\caption{Algorithm to simulate from the maximally skewed stable distribution $F(x;1,-1,\pi/2,0)$.  $\Unif(0,1)$ denotes the uniform distribution on $(0,1)$} \label{alg:sim}
\begin{center}
\begin{tabular}{l l}
\hline \\
{\bf 1:} & Simulate $U_1, U_2 \sim \Unif(0,1)$ independently. \\
{\bf 2:} & Let $W_1 = \pi (U_1 - \frac{1}{2})$ and $W_2 = -\log U_2$. \\
{\bf 3:} & Return $\tan(W_1) [\frac{\pi}{2} - W_1] + \log \left ( W_2  \frac{\cos W_1}{\pi/2 - W_1} \right ) $. \\
\\
\hline 
\end{tabular}
\end{center}
\end{table}   

Randomized algorithms for data sketching are probabilistic, in the sense that data stream elements are mapped deterministically to copies of pseudo-random variables, and the variables are transformed to form a synopsis representation of the data stream.  From this representation, an estimate of the Shannon entropy is derived whose accuracy can be guaranteed, to within a specified level $\epsilon$, with probability exceeding $1-\rho$.  In particular, a randomized algorithm for estimating $H(p)$ will return an $(\epsilon, \rho)$-approximation $\hat{H}(p)$ that satisfies: $\P (|\hat{H}(p) - H(p)| \leq \epsilon H(p) ) \geq 1 - \rho$ for a multiplicative approximation, and $\P (|\hat{H}(p) - H(p)| \leq \epsilon ) \geq 1-\rho$ for an additive approximation.  

In fact, it suffices to have a randomized algorithm that returns an approximation $\hat{H}(p)$ to within accuracy $\epsilon$ with probability greater than $0.5$.  An application of Chernoff's bounds \citep{Hoeffding.63} shows that from $n=\log (1/\rho)$ independent repetitions of the algorithm, the median of $\hat{H}_1(p), \ldots, \hat{H}_n(p)$ is an $(\epsilon, \rho)$-approximation of $H(p)$, where $\hat{H}_i(p)$ is the approximation from the $i$th repetition.  Hence, in general we speak of $\epsilon$-additive and $\epsilon$-multiplicative approximations.

\subsection{Related work on Shannon entropy estimation} \label{sec:related}

Approximating the empirical Shannon entropy from estimates of R\'enyi or Tsallis entropies started with the work of \citet{ZLOSWX.07}.  These authors show that the function $x \log(x)$ can be well approximated by a linear combination of two functions of the form $x^p$, $p \in (0,2]$, for $x$ less than an upper bound.  Summing over distinct data types, they obtain an estimate of the entropy from a linear combination of two frequency moments, effectively interpolating the entropy from two distinct values of $S_{\alpha}(p)$.   \citet{ZLOSWX.07} estimate the frequency moments by the random projections method of \citet{Indyk.06}.  For data types whose cumulative quantity exceeds the upper bound, they estimate the contribution to the entropy separately.

More generally, \citet{HNO.08.2} estimate the Shannon entropy via interpolation from several Tsallis entropy estimates, computed at optimal values of $\alpha$ to minimize the approximation error.   For arbitrary accuracy parameter $\epsilon > 0$, they present additive and multiplicative approximations in space $O(\epsilon^{-2} \log T (\log \log T + \log(1/\epsilon))^{O(1)} )$ for the relaxed strict-turnstile model.  The multiplicative approximation algorithm has near-optimal space complexity in terms of its dependence on $\epsilon$, compared to the lower bound of $\Omega(\epsilon^{-2} / \log^2(\epsilon^{-1}))$ \citep{CCM.10}.  Extensions to the general update model with no restrictions on deletions are possible, but the space bounds increase typically by a factor of $O(\log T)$ \citep{HNO.08.2}. 

Li presents two estimators (the geometric and harmonic mean estimators) of the $\alpha$th frequency moment, based on random projections to the symmetric, $\alpha$-stable distribution \citep{Li.08}, or the positive, $\alpha$-stable distribution \citep{Li.09a, Li.09b}.  The latter method is called {\it compressed counting}, and it improves over symmetric stable random projections in terms of the asymptotic variance of the estimator around $\alpha = 1$.  Compressed counting was recently applied to estimate the cardinality of a data stream in \cite{CC.12}.  \citet{Li.09a, Li.09b} suggests that the Shannon entropy can be estimated from expressions $S_{\alpha}(p)$ or $H_{\alpha}(p)$ with $\alpha$ close to 1 using the geometric mean estimator.  Unfortunately, the resulting algorithm is impractical since it has complexity of order $O(1/\Delta)$ with $\Delta=1-\alpha$.  \citet{LZ.11} offer a marked improvement with a new compressed counting algorithm that provides an $\epsilon$-additive estimate of the Shannon entropy with complexity $O(1/\epsilon^2)$.  In particular, they estimate the Shannon entropy with $\alpha \approx 1$ by
$$
H_{\alpha}(p) = -\log J_{\alpha}(p) - \frac{1}{\Delta} \log F_{1}^{\alpha}(p),
$$
where $J_{\alpha}(p) = F_{\alpha}^{-1/\Delta}(p)$, and the first frequency moment is computed exactly.  Moreover, the estimator of $J_{\alpha}(p)$ has near-optimal efficiency properties in estimating $F_{\alpha}(p)$, and exponentially decreasing tail bounds.  

In addition, \citet{HNO.08} analyze the rate of convergence of the R\'enyi entropy estimate of $H_{\alpha}(p)$ to the Shannon entropy as $\alpha \to 1^{+}$, and provide an explicit formula for $\alpha >1$ that guarantees an $\epsilon$-additive approximation.  They estimate $F_{\alpha}(p)$ by the method of symmetric stable random projections \citep{Li.08}, where $0 < \alpha \leq 2$.  However, a value of $\alpha$ exceeding 1 is not appropriate for maximally skewed stable random projections that require $\alpha < 1$.  Another disadvantage is the prohibitively large space complexity, of order $\tilde{O}(\epsilon^{-4} \log^4 T)$, ignoring logarithmic terms.

The problem of estimating the Shannon entropy is related to that of measuring pairwise independence via the Kullback-Leibler divergence \citep{KL.51}, where the latter has received a lot of attention in recent literature \citep{IM.08, GIM.08, BO.10}.  \citet{IM.08} present a single-pass algorithm for an $\epsilon$-additive approximation of the mutual information between two data streams.  The empirical mutual information is the Kullback-Leibler divergence of the joint distribution and the product of the marginals.  Their algorithm has space complexity $\tilde{O}(\epsilon^{-2})$, but an $(1+\epsilon)$-multiplicative approximation of the mutual information does not exist in small space \citep{IM.08}.   Since the mutual information can be expressed as the sum of the empirical Shannon entropies of the marginals minus the empirical Shannon entropy of the joint, our estimation approach can provide an additive approximation.   The same holds for the conditional entropy represented in terms of Shannon entropies.

\subsection{Our contributions} Let $\delta$ denote $-H(p)$, the negative of the empirical Shannon entropy.  In Section~\ref{sec:der} we present a family of log-mean estimators, denoted by $\hat{\delta}_{lm}(\zeta)$ and indexed by $\zeta > 0$, for the additive approximation of $\delta$.  The algorithm in Table~\ref{alg:est} implements the estimation procedure with $\zeta=1$, and has the following properties:
\begin{itemize}
\item
It requires a single pass over the data stream.
\item
It constructs a $k$-dimensional data sketch by projecting to maximally skewed stable random variables with distribution $F(x;1,-1,\pi/2,0)$.  
\item
It returns the log-mean estimator that avoids the problem of parameter calibration with respect to the index $\alpha$ \citep{HNO.08.2, LZ.11}, by going directly to the limit with $\alpha = 1$.   
\end{itemize}
Section~\ref{sec:motivate} provides the motivation, and the details are in Lemmas~\ref{stablelemma} and \ref{sum_stable}.  The proposed estimator with $\zeta = 1$ has the following properties:
\begin{itemize}
\item
It is asymptotically unbiased as $k \to \infty$, and we show in an empirical study that it has good small-sample performance.
\item
By estimating the entropy directly, rather than the R\'{e}nyi entropy $H_{\alpha}(p)$ with $\alpha \approx 1$, we can make precise statements about the efficiency of our estimator: it is near-optimal with asymptotic relative efficiency (ARE) \citep{Lehmann.98} of 0.932.
\item
Lemma~\ref{lemma:tail_bounds} shows that the estimator has exponentially decreasing tail bounds; in particular, for arbitrary $\epsilon > 0$ and fixed $\zeta \leq 1$,
\begin{align*}
\P \left ( \hat{\delta}_{lm}(\zeta) - \delta \geq \epsilon \right ) &< \exp \left ( -k \frac{\epsilon^2}{G_R} \right ) \\
\P \left ( \hat{\delta}_{lm}(\zeta) - \delta \leq -\epsilon \right ) &< \exp \left ( -k \frac{\epsilon^2}{G_L} \right ),
\end{align*}
where $G_L$ and $G_R$ are small constants that tend in value to 6 as $\epsilon \to 0$.  For $\epsilon \in [0.1,1]$, numerical approximations show that these constants fall in ranges $(4.0,6.0)$ and $(6.0, 9.5)$, respectively.  
\item
It follows from Lemma~\ref{lemma:tail_bounds} that for fixed $\rho \in (0,1)$, the data sketch size $k$ must be of order $O(1/\epsilon^2)$.
\item
The space complexity of the algorithm is $O \left (1/\epsilon^{2} \log T \log(T/\epsilon) \right )$ bits of space, which is near-optimal in terms of dependence on $\epsilon$; in particular, it is optimal up to $\log(1/\epsilon)$ \citep{KNPW.11}.
\end{itemize}  

\begin{table}[!ht]
\caption{Algorithm to approximate the empirical Shannon entropy of a data stream $\mathcal{S}_T$ via the log-mean estimator $\hat{\delta}_{lm}(1)$} \label{alg:est}
\begin{center}
\begin{tabular}{l l}
\hline \\
\hspace{0.2cm}{\bf 1:} & Initialize data sketch $(y_1, \ldots, y_k) = (0, \ldots, 0)$. \\
\hspace{0.2cm}{\bf 2:} & Set the counter $Y = 0$. \\  
\hspace{0.2cm}{\bf 3:} & For $t = 1$ to $T$ \\
\hspace{0.2cm}{\bf 4:} & \quad Update the counter $Y = Y + d_t$. \\
\hspace{0.2cm}{\bf 5:} & \quad Seed the PRNG with $i_t$. \\
\hspace{0.2cm}{\bf 6:} & \quad For $j = 1$ to $k$ \\
\hspace{0.2cm}{\bf 7:} & \quad \quad Generate $R_j(i_t) \sim F(x;1,-1,\pi/2,0)$ \\
\hspace{0.2cm}{\bf 8:} & \quad \quad Update $y_j = y_j + R_j(i_t) \times d_t$. \\
\hspace{0.2cm}{\bf 9:} & At time $t=T$, set $y_j = y_j / Y$ for $j=1, \ldots, k$. \\
{\bf 10:} & Return $\hat{H}(p)  = - \log \left ( k^{-1} \sum_{j=1}^k \exp(y_j) \right )$. \\
\\
\hline
\end{tabular} 
\end{center}
\end{table}

\section{THE LOG-MEAN ESTIMATOR}

\subsection{The method of random projections} \label{sec:motivate}

The method of random projections requires that each element type $c_j \in \mathcal{D}$ that appears in the data stream can be transformed into a distinct random variable $R(c_j)$.  In practice, this is achieved ``to adequate approximation'' by the method of seeding as follows: (i) map $c_j$ to an integer (or vector of integers), (ii) use these integers to seed a pseudo-random number generator (PRNG), and (iii) use the seeded PRNG to simulate the random variable $R(c_j)$.  

\citet{Nisan.92} shows that there exists an explicit implementation of a PRNG that converts a random seed to a sequence of bits, indistinguishable from truly random bits.  So we assume that our PRNG produces truly random variables $R(c_j)$.

The projection is then accumulated online as $\sum_{t=1}^TR(i_t) d_t = \sum_{j=1}^N R(c_j) a_j$. This sum is the dot product of the vector of cumulative quantities with a vector of $N$ independent random variables, each drawn from the maximally skewed stable distribution with $\alpha = 1$. This provides a single element of the data sketch.  A further $k-1$ elements are generated independently in parallel to form the $k$-dimensional data sketch.

We now motivate the use of the maximally skewed stable distribution with $\alpha = 1$ in the random projections method by showing how the problem of estimating the Shannon entropy reduces to that of approximating a location parameter. 

Define the quantity 
$$
B_\alpha = \left (\sum_{j=1}^N p_j^\alpha \right )^{1/\alpha} = F_{\alpha}^{1/\alpha}(p).
$$ 
Let 
$$
Z_\alpha \sim F \left (z; \alpha, 1, \left (\cos \left (\frac{\pi \alpha}{2} \right ) \right )^{1/\alpha}, 0 \right ),
$$ 
for fixed $0 < \alpha < 1$; this is the positive, strictly stable distribution with Laplace transform $e^{-\lambda^{\alpha}}$ for $\lambda \geq 0$.  Let $\left (Z_\alpha^{(1)},\dots,Z_\alpha^{(N)} \right )$ be a vector of independent copies of $Z_\alpha$ and let $p = (p_1,\dots,p_N)$ be a vector of frequencies that satisfy $\sum_{j=1}^N p_j = 1$. From \citet{Zolotarev.86}, we have that 
\begin{equation}\label{basic}
\sum_{j=1}^N Z_\alpha^{(j)} p_j \sim Z_\alpha \left(\sum_{j=1}^N p_j^\alpha\right)^{1/\alpha} = Z_\alpha B_\alpha.
\end{equation}
Projecting to the positive, strictly stable distribution and maintaining weighted linear combinations as in \eqref{basic} is precisely the method of compressed counting \citep{Li.09a, Li.09b}.  Compressed counting reduces the problem of Shannon entropy estimation to that of estimating the scale parameter $B_{\alpha} = J_{\alpha}^{-\Delta/\alpha}(p)$.

Instead, we project to the maximally skewed stable distribution with $\alpha=1$ and $\beta=-1$, defined on the entire real axis.  Starting from the R\'enyi entropy 
$$
H_{\alpha}(p) = \frac{\alpha}{1-\alpha} \log B_{\alpha},
$$
it is easy to show that as $\alpha \to 1$,
$$
\frac{1- B_\alpha}{1-\alpha} = \frac{1}{1-\alpha} \left [ 1 - e^{ (1-\alpha)H_{\alpha}(p) / \alpha } \right ] \to \delta,
$$ 
where $-\delta = - \sum_{j=1}^N p_j \log p_j$ is the Shannon entropy.  Next, we define 
\begin{displaymath}
Y_\alpha^{(j)} = \frac{ 1-Z_\alpha^{(j)} }{1-\alpha} + \log(1-\alpha),
\end{displaymath} 
and, using \eqref{basic}, we obtain
\begin{align} \nonumber 
\sum_{j=1}^N Y_{\alpha}^{(j)} p_j &= \sum_{j=1}^N\left[\frac{1 -Z_\alpha^{(j)}}{1-\alpha} + \log(1-\alpha)\right] p_j \\ \label{transform}
&\sim \left[ \frac{1 - Z_\alpha}{1-\alpha} + \log(1-\alpha)\right] + Z_\alpha \frac{(1 - B_\alpha)}{1-\alpha}.
\end{align}
Taking limits, $\sum_{j=1}^N Y_1^{(j)} p_j \sim Y_1 + \delta$, provided $Y_\alpha$ has a proper limit as $\alpha \to 1$, and using the fact that $Z_{\alpha} \to 1$ as $\alpha \to 1$. The following lemma provides the details.  

\begin{lemma}\label{stablelemma}
The random variable $Y_\alpha$ has a proper limit $Y_1$ as $\alpha \to 1$. The variable $Y_1$ has a maximally skewed stable distribution with $\alpha = 1$, and c.f.
$$
\phi(\theta) = \exp \left (-\frac12 \pi|\theta| + i \theta \log|\theta| \right ) = (i \theta)^{i \theta},
$$ 
i.e., $Y_1 \sim F(y;1,-1,\pi/2,0)$.  Moreover, the $k$th moment of the random variable $\exp(Y_1)$ is $k^k$ for all $k>0$.
\end{lemma}
\begin{proof}
See the supplementary material.
\end{proof}

The heart of our algorithm is contained in the following result; it shows that by sketching to the $F(y;1,-1,\pi/2,0)$ distribution, the negative of the Shannon entropy is recovered as the location parameter of the distribution of a linear combination weighted by the empirical probability mass function.  
 
\begin{lemma} \label{sum_stable}
Let $X_1,\ldots,X_N \sim F(x;1,-1,\pi/2,0)$ i.i.d., and let $p_1,\ldots,p_N$ be positive constants satisfying $\sum_{j=1}^N p_j =1$.  Then,
$$
\sum_{j=1}^N p_j X_j \sim F \left (x; 1, -1, \frac{\pi}{2}, \sum_{j=1}^N p_j \log p_j \right ).
$$
\end{lemma} 
\begin{proof}
The c.f.\ of $\sum_{j=1}^N p_j X_j$, for $\theta \in \mathbb{R}$, is given by
\begin{align*}
&\E \exp \left ( i \theta \sum_{j=1}^N p_j X_j \right ) = \prod_{j=1}^N \E \exp \left ( i \theta p_j X_j \right ) \\
& = \prod_{j=1}^N \exp \left ( -\frac{\pi}{2} p_j |\theta| + i \theta p_j \log | \theta p_j | \right ) \\
& = \exp \left ( \frac{\pi}{2} \left [ -|\theta| + i \theta \log |\theta| \frac{2}{\pi} \right ] + i \theta \sum_{j=1}^N p_j \log p_j \right ).
\end{align*}
The first equality follows from properties of characteristic functions of sums of independent random variables, the second follows from Lemma~\ref{stablelemma}, and the third equality uses the fact that $\sum_{j=1}^N p_j = 1$.  Since the distribution of a random variable is specified by its characteristic function \citep{GS.01}, the result follows by comparison to expression $\phi(\theta)$ in Section~\ref{sec:stable}.
\end{proof}

\subsection{Derivation of the family of log-mean estimators} \label{sec:der} 

\begin{lemma} \label{lemma:logmean}
Let $y_1,\ldots,y_k$ be independent samples from the $F \left (y;1,-1,\pi/2, \delta \right )$ distribution, and let $\zeta > 0$ be a constant.  The log-mean estimator of $\delta$ is 
\begin{displaymath}
\hat{\delta}_{lm}(\zeta) = \zeta^{-1} \log \left ( \zeta^{-\zeta} k^{-1} \sum_{j=1}^k \exp(\zeta y_j) \right ).
\end{displaymath}
As the sample size $k$ increases to $\infty$, the estimator is asymptotically unbiased; in particular, as $k \to \infty$,
$$
\sqrt{k} \left (\hat{\delta}_{lm}(\zeta) - \delta \right ) \to \Normal \left ( 0, \frac{4^{\zeta} - 1}{\zeta^2} \right ).
$$
Moreover, the Fisher information about $\delta$ contained in a single random variable from the $F \left (y;1,-1,\pi/2, \delta \right )$ distribution is approximately $0.3578$, so the ARE of $\hat{\delta}_{lm}(\zeta)$ is 
$$
ARE(\hat{\delta}_{lm}(\zeta)) = \frac{\zeta^2}{0.3578 (4^{\zeta} - 1)}.
$$
Hence, the estimator $\hat{\delta}_{lm}(1.15)$ is near-optimal with largest ARE of 0.942, and $\hat{\delta}_{lm}(1)$ has ARE of 0.932.    
\end{lemma}
\begin{proof}
See the supplementary material.
\end{proof}

In Section~\ref{sec:tails}, we show that the log-mean estimator has exponentially decreasing tail bounds only for $\zeta \leq 1$.  For this reason, we recommend the estimator with $\zeta = 1$ that attains the maximum ARE over the range $\zeta \leq 1$.  The algorithm in Table~\ref{alg:est} returns $\hat{H}(p) = -\hat{\delta}_{lm}(1)$.

\section{PERFORMANCE, TAIL BOUNDS, AND SPACE COMPLEXITY} \label{sec:tails}
Figure~\ref{figure:ARE_entropy} compares the log-mean estimator $\hat{\delta}_{lm}(\zeta)$ to the R\'enyi entropy estimator $\hat{H}_{\alpha}(p)$ \citep{LZ.11} in terms of asymptotic relative efficiency over the range of values $\zeta \in (0,2]$ and $\alpha \in [0.95, 0.99]$.   The Fisher information about $H_{\alpha}(p)$ contained in a positive, $\alpha$-stable random variable with scale parameter $F_{\alpha}^{1/\alpha}(p)$ is given by the expression $\I_2(\alpha) -1$ in \citet{LZ.11}.  The solid, straight line in Figure~\ref{figure:ARE_entropy} is the ARE of $\hat{H}_{\alpha}(p)$, given by 
$$
ARE(\hat{H}_{\alpha}(p)) = \frac{1}{(\I_2(\alpha )- 1) (1 + 2\alpha)},
$$
since $\sqrt{k} (\hat{H}_{\alpha}(p) - H_{\alpha}(p) ) \to \Normal(0, 1+2\alpha)$ as $k \to \infty$ \citep{LZ.11}.

Figure~\ref{figure:MSE_entropy} compares the performances of the log-mean estimator with $\zeta=1, 1.5$ and the estimator $\hat{H}_{\alpha}(p)$ with $\alpha = 0.97$ in terms of relative mean square error (MSE) in small samples.  Plotted for comparison is the Cram\'er-Rao lower bound, defined by $(k \times 0.3578)^{-1}$, and giving a lower bound on the variance of any unbiased estimator of $\delta$.  The MSE is given by $\mathbb{E} (\hat{\delta}_{lm}(\zeta) - \delta )^2$ for the log-mean estimator, and by $\mathbb{E} (\hat{H}_{\alpha}(p) + \delta )^2$ for the R\'enyi estimator, since $\hat{H}_{\alpha}(p)$ estimates $H(p) = -\delta$ for $\alpha \approx 1$.  The expectation is estimated from $10^5$ replicates.

\begin{figure}[!ht] 
\centerline{
\includegraphics[scale=0.4,angle=270]{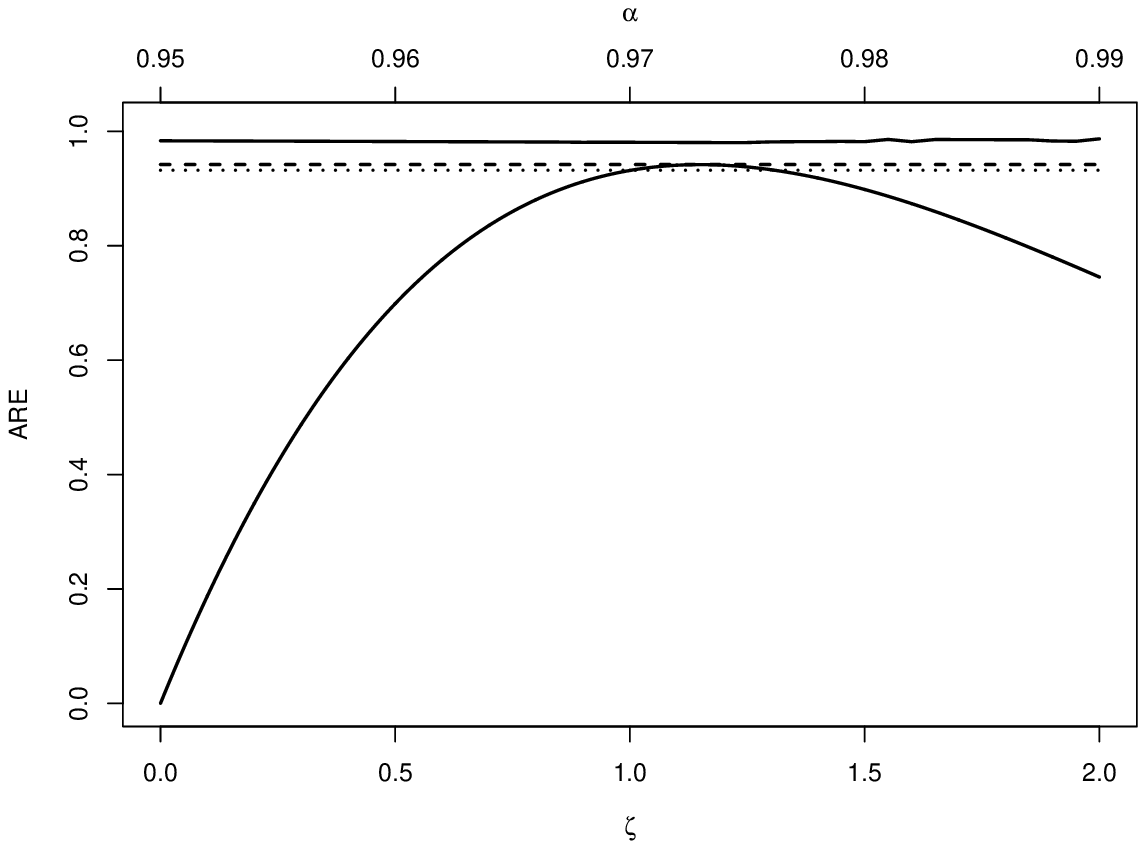}}
\caption{Comparison in terms of asymptotic relative efficiency of the log-mean estimator $\hat{\delta}_{lm}(\zeta)$ for $\zeta \in (0,2]$ (curved, solid line, bottom axis) to the R\'enyi entropy estimator $\hat{H}_{\alpha}(p)$ for $\alpha \in [0.95,0.99]$ (straight, solid line, top axis).  Horizontal lines are drawn at ARE = 0.942 (long dashed line) and ARE=0.932 (dotted line), and a vertical line at $\zeta = 1$, or equivalently, $\alpha=0.97$.    The ARE of $\hat{\delta}_{lm}(\zeta)$ refers to Shannon entropy estimation, whereas the ARE of $\hat{H}_{\alpha}(p)$ refers to R\'enyi entropy estimation.}
\label{figure:ARE_entropy} 
\end{figure}

All computations were performed using the statistical software {\tt R} (http://www.r-project.org/), and the same string of random numbers was employed in the computations of each estimator.  The estimator $\hat{H}_{\alpha}(p)$ is derived from a random sample of positive, $\alpha$-stable random variables, raised to the power $-\alpha/\Delta$ and scaled by $F_{\alpha}(p)^{-1/\Delta}$.   With $\alpha \approx 1$, $F_{\alpha}(p)$ is approximately equal to $\sum_{i=1}^N a_i$, the total cumulative quantity, and, if this quantity is large, then the scaling factor is effectively zero.  We experience this problem for our simulated data stream with values of $\alpha \geq 0.98$, hence our choice of $\alpha = 0.97$ for the small-sample comparison in Figure~\ref{figure:MSE_entropy}.   The estimator $\hat{H}_{0.97}(p)$ has ARE of 0.981. 

Figure~\ref{figure:MSE_entropy} shows that the three estimators have comparable performance in small samples.  Compared to the Cram\'er-Rao lower bound, the performance of the log-mean estimators is particularly good for $k \geq 20$.  However, we do not have a corresponding lower bound for the performance of $\hat{H}_{\alpha}(p)$ as an estimator of the Shannon entropy, since an analysis of the rate of convergence of the R\'enyi entropy to the Shannon entropy as $\alpha \to 1^{-}$ is lacking.

\begin{figure*}[!ht] 
\vspace{.3in}
\centerline{\includegraphics[scale=0.5,angle=270]{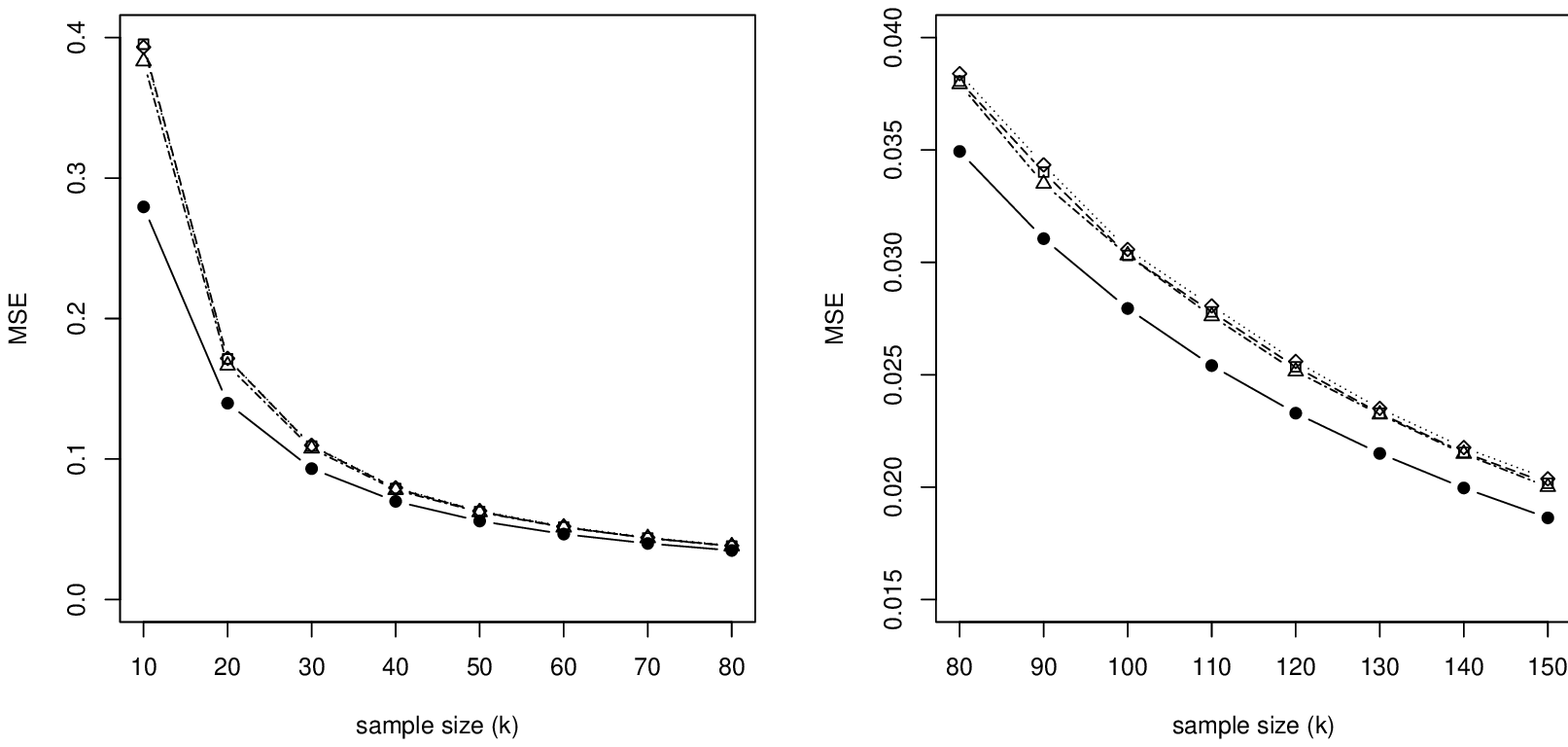}}
\vspace{.3in}
\caption{Comparison in terms of MSE of the log-mean estimators $\hat{\delta}_{lm}(1)$ (dotted line, $\Diamond$),  $\hat{\delta}_{lm}(1.15)$ (long dash line, $\Box$), and the estimator $\hat{H}_{0.97}(p)$ (two dash line, $\triangle$) of \cite{LZ.11}.  The solid line is the Cram\'er-Rao lower bound on the variance of an unbiased estimator of the Shannon entropy.  The MSE is estimated from $10^5$ replicates.  The MSE lines are indistinguishable, and, for $k \geq 20$, the small-sample performance of the log-mean estimators is very good compared to the Cram\'er-Rao lower bound.}
\label{figure:MSE_entropy} 
\end{figure*}

The length of the data sketch vector, $k$, is determined by the behaviour of the tail bounds of the additive approximation error.  Lemma~\ref{lemma:tail_bounds} shows that for $\zeta \leq 1$, the log-mean estimator has exponentially decreasing tail bounds. 

\begin{lemma} \label{lemma:tail_bounds}
Exponentially decreasing tail bound exist for $\zeta \leq 1$ and arbitrary $\epsilon > 0$, with
\begin{align*}
\P \left ( \hat{\delta}_{lm}(\zeta) - \delta \geq \epsilon \right ) &< \exp \left ( -k \frac{\epsilon^2}{G_R} \right ), \\
\P \left ( \hat{\delta}_{lm}(\zeta) - \delta \leq -\epsilon \right ) &< \exp \left ( -k \frac{\epsilon^2}{G_L} \right ),
\end{align*}
where 
$$
G_R = \frac{\epsilon^2}{\sup_{t > 0} Q_\zeta(t,\epsilon)} ,\ G_L = \frac{\epsilon^2}{\sup_{t > 0} Q_\zeta(-t,-\epsilon)},
$$ and
$$
Q_\zeta(t,\epsilon) = -\log \left (\sum_{j=0}^{\infty} t^j \frac{j^{\zeta j}}{j!} \right ) + t e^{\zeta \epsilon}.
$$
Furthermore as $\epsilon \to 0$ both $G_R$ and $G_L$ tend to $2(4^\zeta-1)/\zeta^2$.
\end{lemma}

\begin{proof}
See the supplementary material.
\end{proof}

Given $\epsilon > 0$ and $0 < \rho < 1$, bounding the additive approximation error by
$$
\P \left ( | \hat{\delta}_{lm}(1) - \delta | \geq \epsilon \right ) < \rho,
$$
requires that 
$$
k > -\frac{G}{\epsilon^2} \log \left ( \frac{\rho}{2} \right ) = O \left ( \frac{1}{\epsilon^2} \right ),
$$
where $G = \max \{G_L, G_R \}$.  For $\epsilon \in  [0.1, 1]$, numerical approximations show that the constants $G_R$ and $G_L$ fall in ranges (4.0, 6.0) and (6.0, 9.5), respectively, so the hidden constant in the big O notation, ignoring the term $\log(1/\rho)$, is small, for small $\epsilon$.  Hence, the algorithm requires $O(\epsilon^{-2} \log T)$ random bits of space for a data stream with $d_t \in \{-1,1\}$.  The space complexity increases to $O \left (\epsilon^{-2} \log T \log(T/\epsilon) \right )$ bits after applying Nisan's PRNG \citep{Nisan.92, Indyk.06}. In the general case that allows insertions and deletions with $d_t \in \{-M,\ldots,M\}$, it suffices to increase $T$ by a factor $M$ \citep{HNO.08.2}.

\section{CONCLUSION}
This paper joins a growing body of literature on estimating the empirical Shannon entropy over streaming data efficiently, with small memory usage and fast updates.  In particular, we adopt the method of random projections to the maximally skewed, strictly stable distribution with parameters $\alpha = 1$ and $\beta=-1$, thus avoiding the problem of the choice of parameter $\alpha$ \citep{HNO.08.2,LZ.11}.  We derive properties of this distribution, showing that it has a surprisingly simple characteristic function $(i \theta)^{i \theta}$ and that the $k$th moment of the exponential of such a variable is $k^k$ for all positive real values of $k$. These properties enable the Shannon entropy to be estimated directly from the associated data sketch as the logarithm of a simple average. 

We recommend the asymptotically unbiased log-mean estimator with $\zeta = 1$ to provide an additive approximation of the Shannon entropy.   By estimating the entropy directly, rather than via the R\'enyi entropy with $\alpha \approx 1$, we can determine the asymptotic relative efficiency of our estimator: 0.932 with $\zeta = 1$.  Moreover, the probability of the estimator having an additive error greater than $\epsilon$ decreases exponentially with $k\epsilon^2$ for small $\epsilon$, where $k$ is the size of the data sketch.  This results in a near-optimal space complexity bound of \begin{math}
O \left (\epsilon^{-2} \log T \log(T/\epsilon) \right ),
\end{math}
where $T$ is the length of the data stream observed.  


\subsubsection*{References}
\bibliographystyle{plainnat}
\bibliography{entropy_article}

\newpage
\subsubsection*{APPENDIX - SUPPLEMENTARY MATERIAL}
\paragraph{Proof of Lemma~\ref{stablelemma}}

Following \cite{Zolotarev.86}, the c.f.\ of $Z_\alpha$ equals
$$
\mathbb{E} e^{i \theta Z_{\alpha}} = \exp\left\{-|\theta|^\alpha\cos \left ( \frac{\pi \alpha}{2} \right ) + i |\theta|^{\alpha} \text{sgn}(\theta) \sin \left ( \frac{\pi \alpha}{2} \right )\right\},
$$
for $\theta \in \mathbb{R} $, where $\text{sgn}(\theta) = \theta/|\theta|$ for $\theta \neq 0$, and 0 otherwise. As $\alpha \to 1$, $\mathbb{E} e^{i \theta Z_{\alpha}} \to e^{i \theta}$, so $\lim_{\alpha \to 1} Z_{\alpha} = 1$.  It follows that the limit $Y_1 = \lim_{\alpha \to 1}Y_{\alpha}$ exists. 

 For $\theta \in \mathbb{R}$, the c.f. of $Y_{\alpha}$ is given by 
\begin{align} \nonumber
\mathbb{E} e^{ i \theta Y_{\alpha}} & = \exp \left \{ i \theta \left ( \frac{1}{1-\alpha} + \log(1-\alpha) \right ) \right \} \mathbb{E} e^{ -i \theta \frac{Z_{\alpha} }{1-\alpha}} \\ \nonumber
& =- \left | \frac{\theta}{1-\alpha} \right |^{\alpha}  \exp \left \{ i \theta \left ( \frac{1}{1-\alpha} + \log(1-\alpha) \right )\right \} \\ \label{cf}
& \qquad \times \exp \left \{ i \text{sgn}(\theta) \frac{\pi \alpha}{2} \right \}
\end{align}
Letting $\alpha \to 1$ in \eqref{cf}, we have the desired limit $(i \theta)^{i \theta}$.  Lastly, the $k$th moment of $\exp (Y_{\alpha})$ for $k > 0$ is given by
\begin{align} \nonumber 
\mathbb{E} e^{ k Y_{\alpha}} &= (1-\alpha)^{k} e^{k/(1-\alpha)} \mathbb{E} e^{ -k Z_{\alpha}/(1-\alpha) } \\ \label{mgf}
&= (1-\alpha)^{k} e^{- [k/(1-\alpha)]^{\alpha} + k/(1-\alpha)},
\end{align}
where the second equality follows from the Laplace transform of $Z_{\alpha}$.  Taking the limit as $\alpha \to 1$ in \eqref{mgf}, we obtain the $k$th moment of $\exp(Y_1)$ for $k > 0$ as follows.  Define $n = 1/(1-\alpha)$.
\begin{align*}
\E e^{ kY_{\alpha} } &= \exp \left \{ k \left [ n - \log n - \frac{n}{(kn)^{1/n}} \right ]  \right \} \\
&= \exp \left \{  - k n^{-1/n} \left ( n \left [ k^{-1/n} -1 \right ] \right ) \right \} \times \\
& \qquad \exp \left \{ k \left [n - n^{1-1/n} - \log n \right ] \right \}.
\end{align*}
As $\alpha \to 1$, $n \to \infty$, and we have that $\lim_{n \to \infty} n^{-1/n} = 1$ and $\lim_{n \to \infty} n \left [ k^{-1/n} -1 \right ] = \log(1/k)$.  It remains to show that $n - n^{1-1/n} - \log n \to 0$ as $n \to \infty$.  By rewriting
$$
n - \frac{n}{n^{1/n}} - \log n = n \left [  1 - 1/n^{1/n} + \log \left (1/n^{1/n} \right ) \right ], 
$$
we use the fact that for $n > 1$, the Taylor expansion of $\log \left (1/n^{1/n} \right )$ is
\begin{align*}
\log \left (n^{-1/n} \right ) &= \left ( n^{-1/n} - 1 \right ) + \sum_{i=2}^{\infty} \frac{(-1)^{i+1}( n^{-1/n} - 1)^i}{i}  \\
&= \left ( n^{-1/n} - 1 \right ) + O \left ( (n^{-1/n} - 1)^2 \right ).
\end{align*}
So,
$$
n - \frac{n}{n^{1/n}} - \log n = O \left ( n^{1-2/n} - 2n^{1-1/n} + n \right ),
$$
and the right hand side converges to 0 as $n \to \infty$.

\paragraph{Proof of Lemma~\ref{lemma:logmean}}

Consider the following transformation:
$$
w_j = e^{\zeta y_j} = e^{\zeta(\delta + z_j)} = e^{\zeta \delta} e^{\zeta z_j},
$$
where $z_j \sim F (z;1,-1,\pi/2,0 )$ i.i.d.\ have characteristic function $\phi(\theta) = \mathbb{E}\exp(i \theta z_j) = (i\theta)^{i\theta}$, for $\theta \in \mathbb{R}$.  Then, from Lemma \ref{stablelemma},
\begin{math}
\mathbb{E}w_j = e^{\zeta \delta} \zeta^{\zeta}.
\end{math}
Let $\eta = e^{\zeta \delta}$.  The estimator 
$$
\hat{\eta}(\zeta) = \zeta^{-\zeta} k^{-1} \sum_{j=1}^k w_j
$$ 
is unbiased for $\eta$, i.e., $\E \hat{\eta}(\zeta) = \eta$, and has variance $\v (\hat{\eta}(\zeta)) = \eta^2 k^{-1} (4^{\zeta} - 1)$.  

Moreover, by the Central Limit Theorem, as $k \to \infty$,
$$
\sqrt{k} \eta^{-1} (\hat{\eta}(\zeta) - \eta) \to \Normal \left (0, 4^{\zeta} - 1 \right ).
$$
\noindent
The log-mean estimator of $\delta$ is
$$
\hat{\delta}_{lm}(\zeta) = \zeta^{-1} \log \hat{\eta} = \zeta^{-1} \log \left ( \zeta^{-\zeta} k^{-1} \sum_{j=1}^k \exp(\zeta y_j) \right ).
$$
\noindent
By the Delta Method, as $k \to \infty$,
\begin{math}
\sqrt{k} \big (\hat{\delta}_{lm}(\zeta) - \delta \big ) \to \Normal \big ( 0, \zeta^{-2}(4^{\zeta} - 1) \big ),
\end{math}
so $\hat{\delta}_{lm}(\zeta)$ is asymptotically unbiased for $\delta$.

Finally, we want to find the optimal value of $\zeta$ that maximizes the ARE of $\hat{\delta}_{lm}(\zeta)$ relative to the MLE of $\delta$.  So, we begin by estimating the Fisher information about $\delta$ contained in a single random variable following the $F \left (y;1,-1,\pi/2, \delta \right )$ distribution.  Let $f \left (y;1,-1,\pi/2, \delta \right )$ denote the corresponding density function.  From Algorithm 1, it is possible to show that the density is given by
$$
f \left (y;1,-1,\pi/2, \delta \right ) = \frac{e^{y - \delta}}{\pi} \int_0^{\pi} e^{-g(w)} e^{ -e^{y - \delta - g(w)}} dw,
$$
for $-\infty < y < \infty$, where
$$
g(w) = \frac{w}{\tan(w)} + \log \left ( \frac{\sin(w)}{ w} \right ).
$$
And the Fisher information about $\delta$ is expressed as
\begin{align*}
\I_1(\delta) &= \E \left (  \frac{\partial }{\partial \delta} \log f(y;1,-1,\pi/2, \delta)\right )^2 \\
&= 1 - \frac{2}{\pi} \int_0^{\infty} s \I(2,s)ds + \frac{1}{\pi} \int_0^{\infty} s^2 \frac{\I(2,s)^2}{\I(1,s)}ds,
\end{align*}
where
$$
\I(l,s) = \int_{0}^{\pi} e^{-lg(w)} e^{-s e^{-g(w)}}dw, \quad  l=1,2.
$$

We evaluate the integrals in $\I_1(\delta)$ numerically, and obtain that the Cram\'er-Rao lower bound for estimating $\delta$ is approximately $1/\I_1(\delta) = (0.3578)^{-1}$.  So the ARE is 
$$
\frac{\zeta^2}{0.3578(4^{\zeta} - 1)}.
$$  
This is a concave function that attains a maximum value of $0.942$ when $\zeta \approx 1.15$.  When $\zeta = 1.0$, the ARE evaluates to $0.932$.

\paragraph{Proof of Lemma~\ref{lemma:tail_bounds}}
For $\epsilon > 0$ and $t > 0$, 
\begin{align*}
\P \left ( \hat{\delta}_{lm}(\zeta) - \delta \geq \epsilon \right ) &= \P \left ( \frac{\zeta^{-\zeta}}{k} \sum_{j=1}^k \exp(\zeta z_j)  \geq e^{\zeta \epsilon} \right ) \\
&\leq e^{ -tk e^{\zeta \epsilon} } \mathbb{E} \exp \left \{ \sum_{j=1}^k \frac{t e^{\zeta z_j}}{\zeta^{\zeta}} \right \},
\end{align*}
by the Chernoff bound \citep{GS.01}, provided the right hand side converges.  Define $T_j = t \zeta^{-\zeta} e^{\zeta z_j}$, $j=1,\ldots,k$.  Then, 
\begin{align*}
\mathbb{E} \exp \sum_{j=1}^k T_j &= \left ( \mathbb{E} \exp(T_1) \right )^k = \left \{ \sum_{j=0}^{\infty} \mathbb{E} T_1^j / j! \right \}^k \\
&= \left \{ \sum_{j=0}^{\infty} t^j j^{\zeta j} / j! \right \}^k.
\end{align*}
\noindent
By the Ratio Test, the series is absolutely convergent for all $t > 0$ if $0 < \zeta < 1$, and for $0 < t < e^{-1}$ if $\zeta = 1$.  If $\zeta > 1$, the series is divergent.  Define $T = \{t; t > 0\}$ for $0 < \zeta < 1$, and $T = \{t; 0 < t < e^{-1} \}$ for $\zeta = 1$.  It follows that, if $\zeta \leq 1$, then $\hat{\delta}_{lm}(\zeta)$ has an exponentially decreasing right tail bound that satisfies
$$
\P \left ( \hat{\delta}_{lm}(\zeta) - \delta \geq \epsilon \right ) < \exp \big ( -k \epsilon^2 / G_R \big ),
$$
where
\begin{equation} \label{eq:sup}
\frac{\epsilon^2}{G_R} = \sup_{t \in T} \bigg \{ -\log \Big ( \sum_{j=0}^{\infty} \frac{t^j j^{\zeta j}}{j!} \Big ) + t e^{\zeta\epsilon} \bigg \}.
\end{equation}   
\noindent 
It is straightforward to show that the function maximized in \eqref{eq:sup} is concave.  The result follows similarly for the left tail bound. Furthermore by expanding the series in \eqref{eq:sup} for small values of $t$ we can show that as $\epsilon \to 0$ both $G_R$ and $G_L$ converge to $2(4^\zeta-1)/\zeta^2$.  The details are as follows.

Define
$$
M_{\zeta}(t) = \sum_{j=0}^{\infty} \frac{t^j j^{\zeta j}}{j!},
$$ 
and consider
$$
K_{\zeta}(s, \epsilon) = \big ( M_{\zeta}(s\epsilon) \exp(-s \epsilon e^{\zeta \epsilon}) \big )^{1/\epsilon^2}, \ s>0.
$$
$K_{\zeta}(s, \epsilon)$ is a convex function \citep{GS.01}, so it follows that $\inf_{s > 0}K_{\zeta}(s, \epsilon) \to \inf_{s > 0} K_{\zeta}^{\star}(s)$, where $K_{\zeta}^{\star}(s)$ is the pointwise limit of $K_{\zeta}(s, \epsilon)$ as $\epsilon \to 0$, provided this limit exists.  Furthermore, since $1/G_R = -\log \big (\inf_{s > 0} K_{\zeta}(s, \epsilon) \big )$, it follows that 
\begin{math}
\lim_{\epsilon \to 0} G_R = -\big [\log (\inf_{s > 0} K_{\zeta}^{\star}(s)) \big ]^{-1}.
\end{math}
To establish the pointwise limit, first note that if $s\epsilon \in T$, then
\begin{displaymath}
\sum_{j=3}^{\infty} \frac{(s\epsilon)^j j^{\zeta j}}{j!} \leq \epsilon^3 \sum_{j=3}^{\infty} \frac{s^j j^{\zeta j}}{j!} = o(\epsilon^2).
\end{displaymath}
So that expanding in powers of $\epsilon$, we have that
\begin{align*}
\log(K_{\zeta}(s, \epsilon)) & = \frac{1}{\epsilon^2} \left [ \log \left ( 1 + s\epsilon + \frac{(s\epsilon)^2 4^{\zeta}}{2!} + o(\epsilon^2) \right ) \right ] \\
&  \quad + \frac{1}{\epsilon^2} \left [- s\epsilon(1 + \zeta \epsilon) + o(\epsilon^2) \right ] \\
& = \frac{s^2}{2} \left  \{4^{\zeta} - 1 \right \} = K_{\zeta}^{\star}(s)
\end{align*}
Differentiating with respect to $s$, we obtain
\begin{math}
\inf_{s > 0} K_{\zeta}^{\star} (s) = - \zeta^2 / [2(4^{\zeta} - 1)],
\end{math}
as required, where the convexity ensures a unique minimum.

\end{document}